\RequirePackage{silence}
\RequirePackage{amsmath}
\documentclass[envcountsame,a4paper,12pt]{llncs}
\usepackage[margin=2.3cm,marginpar=1.6cm]{geometry}
\usepackage{multisets}
\usepackage{comment,booktabs,microtype}

\title{Probabilistic team semantics}

\author{Arnaud Durand\inst{1} \and Miika Hannula\inst{2} \and Juha Kontinen\inst{3} \and Arne Meier\inst{4} \and Jonni~Virtema\inst{5}}

\institute{Institut de Math\'ematiques de Jussieu - Paris Rive Gauche, CNRS UMR 7586 - Universit\'e Paris Diderot, \email{durand@math.univ-paris-diderot.fr}
\and Department of Computer Science, University of Auckland, \email{m.hannula@auckland.ac.nz}
 \and Department of Mathematics and Statistics, University of Helsinki, \email{juha.kontinen@helsinki.fi} \and Leibniz Universität Hannover, Institut für Theoretische Informatik, \email{meier@thi.uni-hannover.de}
 \and Databases and Theoretical Computer Science, Hasselt University, \\ \email{jonni.virtema@uhasselt.be} 
 }

\begin{document}
 \maketitle

\begin{abstract}
Team semantics is a semantical framework for the study of dependence and independence concepts ubiquitous in many areas such as  databases and statistics. In recent works team semantics has been generalised to accommodate also multisets and probabilistic dependencies. In this article we study a  variant of probabilistic team semantics and relate this framework to a Tarskian two-sorted logic. We also show that very simple quantifier-free formulae of our logic give rise to $\NP$-hard model checking problems.
 \end{abstract}

\section{Introduction}
Team semantics is the modern approach for the study of logics of dependence and independence. The systematic development of  team semantics began by the introduction of  Dependence Logic in 2007 \cite{vaananen07} although  the key ingredients of the new semantics were already introduced by Hodges 1997 \cite{MR1465612}. In team semantics, satisfaction of formulae is defined not via single assignments but via sets of assignments (teams). Sets of assignments enables one to introduce a multitude of interesting atoms to the logic such as dependence, independence, and inclusion atoms:
 $$\dep(\tuple x, y), \ \tuple y \perp_{\tuple x} \tuple z\, \textrm{ and } \tuple x \subseteq \tuple y$$ 
 that do not make sense with respect to a single assignment. 
Independence logic, introduced by Grädel and Väänänen \cite{gradel13}, extends first-order logic with
independence atoms. The independence atom $\tuple y \perp_{\tuple x} \tuple z$ holds if the value of $\tuple z$ does not tell us anything new about the value of $\tuple y$ when the value of $\tuple x$ is fixed. 
By viewing a team $X$ with domain $\{x_1,\ldots, x_{n}\}$ as a database table over attributes $x_1,\ldots, x_{n}$,   dependence, inclusion, and independence atoms correspond exactly to functional, inclusion, and embedded multivalued dependencies (EMVDs), see, e.g., \cite{DBLP:conf/wollic/KontinenLV13,HannulaKLinki,HannulaK16}. 
Moreover EMVDs and probabilistic conditional independence $\tuple Y \perp\tuple Z|\tuple X$ have significant connections, confer, e.g., \cite{Gyssens2014628,wbw00,CoranderHKPV16}.
Multiteam semantics, introduced by Durand~et~al.~\cite{foiks/0001HKMV16}, is the multiset analogue of team semantics. This setting enables the logical study of probabilistic dependencies such as the probabilistic conditional independence atoms $\pcixyz$ that inherit their semantics from the corresponding notion $\tuple Y \perp\tuple Z|\tuple X$ from statistics.
One of the advantages of multiteam semantics is that it allowes to study the interplay of atoms such as $\dep(\tuple x, y)$, $\tuple y \perp_{\tuple x} \tuple z$, and  $\pcixyz$ in a unified framework.

In this paper, we focus on probabilistic team semantics. A probabilistic team is a set of assignments endowed with a probability distribution that maps each assignment of the set to a ratio.  There is a vast literature on probabilistic logics but so far only few works study probabilistic team semantics. The teams that arise from applications (e.g., database tables) often contain duplicate rows leading naturally to multiteams (i.e., multiset analogues of teams). Furthermore, finite multiteams can be viewed as probabilistic teams endowed with the counting measure induced by the multiplicities. Importantly, in many applications, duplicate rows can store relevant information; e.g., if a table is used to store an outcome of a poll or a collection of outcomes of measurements. In these cases the interest lies in the distribution of the data and not so much in the size of the sample. Hence it makes sense to abstract from the concrete data (multiteams) to the distribution of data (probabilistic teams).
We consider a logic that uses probabilistic independence $\pcixyz$ and
marginal identity atoms $\tuple x \approx \tuple y$ as primitives in the setting of probabilistic team semantics. These atoms were recently introduced by Durand~et~al.~\cite{foiks/0001HKMV16} in the context of multiteam semantics. 
The marginal identity atom $\tuple x \approx \tuple y$ expresses that in a team the distribution of values for the variables $\tuple x$ coincides with that of $\tuple y$.
We relate this logic to a natural variant of (two-sorted) existential second-order logic with quantification over rational distributions. We also consider the complexity of model checking and show that very simple formulae using $\tuple x \approx \tuple y$ give rise to $\NP$-hard model checking problems.  
\begin{example}
Consider a database table that lists results of experiments.  The data can be regarded either as a multiteam or as the related probabilistic team using the counting measure; both interpretations having its own advantages. Each record corresponds to outcomes of measurements obtained simultaneously in two locations. The table has four attributes \textsf{Test1} and \textsf{Test2} that range over the possible types of measurements and \textsf{Outcome1} and \textsf{Outcome2} that range over outcomes of the measurements.
The probabilistic independence atom $\mathsf{Test1} \pperp \mathsf{Test2}$ expresses that the types of measurements are independently picked in the two locations. The marginal identity atom $(\mathsf{Test1},\mathsf{Outcome1}) \approx (\mathsf{Test2},\mathsf{Outcome2})$ expresses that the distributions of results are the same in both test sites. The formula $\mathsf{Test1} = \mathsf{Test2} \lor (\mathsf{Test1} \neq \mathsf{Test2} \land \mathsf{Outcome1} \pperp \mathsf{Outcome2})$ expresses that there is no correlation between outcomes of the different measurements.
\end{example}

\begin{example}
 Consider a database table that describes voting behaviour in two different elections by some sample of voters. Attributes of the table are \textsf{Election1} and \textsf{Election2} that range over political parties. Each record corresponds to a voting behaviour of a voter in the sample. The table then gives rise to a probabilistic team that approximates the voting behaviour of the population. The complex formula $\mathsf{Election1} = \mathsf{Election2} \lor (\mathsf{Election1} \neq \mathsf{Election2} \land \mathsf{Election1} \approx \mathsf{Election2})$ expresses that each party obtained the same portion of swing voters in the second election that it got in the first election.
\end{example}

It is well known that the satisfaction relation of team semantics can be formalised in (existential) second-order logic when the team is encoded by an additional relation. This result gives an upper bound and a ``yardstick'' for the expressive power of many of the logics studied in the team semantics literature. 
One of the motivations for the current article is to develop an analogous yardstick of expressivity for logics over multiteams and probabilistic teams. We use a variant of existential second-order logic over two-sorted structures for this purpose whose first sort encodes the first-order structure and whose second sort consists of the closed interval $[0,1]$ of rational numbers $\rat$ over which arithmetic operations of multiplication and sum can be applied. Distributions from the first sort ranging over the second sort $\rat$ encode probabilistic teams.

In the second part of the article we consider the complexity of model-checking in probabilistic and multiteam  semantics and show that, over multiteams,  very simple formulae using $\tuple x \approx \tuple y$ give rise to $\NP$-hard model checking problems. This  result is in drastic contrast with the influential result of Galliani and Hella \cite{gh13} that inclusion atoms in the ordinary team semantics give rise to a logic equivalent with (a fragment of) the least fixed point logic and accordingly is contained in $\Ptime$.   Interestingly our reduction does not  work under the slightly different probabilistic interpretation of disjunction. It is an open question whether the data-complexity of $\FO(\tuple x \approx \tuple y)$ is in $\Ptime$ for the probabilistic semantics.

\emph{Previous work on probabilistic team semantics:}
Probabilistic versions of dependence logic (and IF-logic) have been previously studied by Galliani, Mann,  Sevenster, and Sandu \cite{Galliani2008,Mann,SevensterS10}. 
 Moreover, Hyttinen et~al.\ \cite{hpv14,aml/HyttinenPV17}  consider so-called quantum team and measure team logics over probabilistic teams and give complete axiomatisation for them.
 It is worth noting, as regards to the connectives and quantifiers, our semantics is similar to the one defined by  Galliani~\cite{Galliani2008} and that the atoms $\pcixyz$ and $\tuple x \approx \tuple y$ were introduced only later by Durand~et~al.~\cite{foiks/0001HKMV16} in the multiteam semantics context.

\section{A variant of existential second-order logic with quantification over rational distributions}
First-order variables are denoted by $x,y,z$ and tuples of first-order variables by $\vec{x},\vec{y},\vec{z}$. The length of the tuple $\vec{x}$ is denoted by $\lvert \vec{x}\rvert$, and for two tuples $\tuple x, \tuple y$  we denote by $\tuple x\setminus \tuple y$ any tuple that lists those elements of $\tuple x$ that do not appear in $\tuple y$. By $\Var(\tuple x)$ we denote the set of variables that appear in the variable sequence $\tuple x$.
A \emph{vocabulary} $\tau$ is a set of relation symbols and function symbols with prescribed arities. We mostly denote relation symbols by $R$ and function symbols by $f$, and the related arities by $\ar(R)$ and $\ar(f)$, respectively.
  A vocabulary is \emph{relational} (resp., \emph{functional}) if it consists of only relation (resp., function) symbols. Similarly, a structure is \emph{relational} (resp., \emph{functional}) if it is defined over a relational (resp., functional) vocabulary. 
We let $\Varfo$ and $\Varso$ denote disjoint countable sets of first-order and function variables (with prescribed arities), respectively. The set of rational numbers in the closed interval $[0,1]$ is denoted by $\rat$. Given a finite set $A$, a function $f\colon A\to\rat$ is called a \emph{(probability) distribution} if $\sum_{s\in A}f(s)=1$. In addition, the empty function is a \emph{distribution}. 

A  relational $\tau$-structure is a  tuple $\A=(A, (R_i^\A)_{R_i\in\tau})$,
where $A$ is a  nonempty set and each $R_i^\A$ is a relation on $A$ (i.e., $R_i^\A\subseteq A^{\ar(R_i)}$). 
In this paper, we consider structures that enrich finite relational $\tau$-structures by adding $\rat$ as a second domain sort and functions that map tuples from $A$ to $\rat$.

\begin{definition}
Let $\tau$ and $\sigma$ be a relational and a functional vocabulary, respectively. A probabilistic $\tau\cup\sigma$-structure is a tuple 
\[
\A = (A,\rat, (R_i^\A)_{R_i\in\tau}, (f_i^\A)_{f_i\in\sigma}),
\]
where $A$ (i.e. the domain of $\A$) is a finite nonempty set, each $R_i^\A$ is a relation on $A$ (i.e., a subset of $A^{\ar(R_i)}$), and each  $f_i^\A$ is a probability distribution from $A^{\ar(f_i)}$ to $\rat$ (i.e., a function such that $\sum_{\vec{a}\in A^{\ar(f_i)}} f_i(\vec{a}) = 1$).
\end{definition}
Note that if $f$ is a  $0$-ary function symbol, then $f^\A$ is
the constant $1$.
Next, we define a variant of functional existential second-order logic with numerical terms ($\ESOf$) that is designed to describe properties of the above probabilistic structures. As first-order terms we have only first-order variables. For a set $\sigma$ of function symbols, the set of numerical $\sigma$-terms $i$ is defined via the following grammar:
 \[
i ::=  f(\vec{x}) \mid i\times i \mid  \SUM_{\vec{x}} i,
\]
where $\vec{x}$ is a tuple of first-order variables from $\Varfo$ and $f\in\sigma$. The value of a numerical term $i$ in a structure $\A$ under an assignment $s$ is denoted by $[i]^\A_s$.
We have the following rules for the numerical terms:
\begin{align*}
&[f(\tuple x)]^\A_s:=f^{\A}(s(\tuple x)), \quad \quad\quad\quad\quad\quad      [i \times j]^\A_s := [i]^\A_s \cdot [j]^\A_s,\\
&[\SUM_{\vec{x}} i(\vec{x},\vec{y})]^\A_s := \sum_{\vec{a}\in A^{\lvert \vec{x} \rvert}} [i(\vec{a},\vec{y})]^\A_s,
\end{align*}
where $\cdot$  and $\sum$ are the multiplication and sum of rational numbers, respectively.
In this context, $i(\vec x,\vec y)$ is a numerical term over variables in $\vec x$ and $\vec y$. 
Note that, in the semantics of $\SUM_{\vec x}i$ the tuple $\vec y$ could be empty.
Furthermore let $\tau$ be a relational vocabulary. The set of $\tau\cup\sigma$-formulae of $\ESOf$ is defined via the following grammar:
\[
\phi ::=   x=y \mid x\neq y \mid i = j \mid i\neq j \mid R(\vec{x}) \mid \neg R(\vec{x}) \mid
 \phi\land\phi \mid \phi\lor\phi \mid \exists x\phi \mid \forall x \phi \mid  \exists f \psi,
 \]
where $i$ is a numerical $\sigma$-term, $R\in\tau$ is a relation symbol, $f\in \Varso$ is a function variable, $\vec{x}$ is a tuple of first-order variables, and $\psi$ is a $\tau\cup(\sigma\cup\{f\})$-formula of $\ESOf$.  Note that the syntax of $\ESOf$ admits of only first-order subformulae to appear in negation normal form. This restriction however does not restrict the expressiveness of the language.

Semantics of $\ESOf$ is defined via probabilistic structures and assignments analogous to first-order logic; note that first-order variables are always assigned to a value in $A$ whereas functions map tuples from $A$ to $\rat$. In addition to the clauses of first-order logic, we have the following semantical clauses:
\begin{align*}
&\A\models_s i=j \Leftrightarrow [i]^\A_s = [j]^\A_s,  \quad\quad\quad\quad\quad \A\models_s i\neq j \Leftrightarrow [i]^\A_s \neq [j]^\A_s,\\
&\A\models_s \exists f \phi \Leftrightarrow \A[h/f]\models_s \phi \text{ for some probability distribution $h\colon A^{\ar(f)} \to \rat$,}
\end{align*}
where $\A[h/f]$ denotes the expansion of $\A$ that interprets $f$ to $h$.

Note that the property of $h$ being a probability distribution can be expressed by the formula $\SUM_{\vec{x}} h(\tuple x)=1$ suggesting that it is not vital whether the restriction to probability distributions is in the semantics or not; in this case, however, $\rat$ would not suffice as a second sort and the set of (non-negative) rationals should be used instead.
Furthermore, for  relating  $\ESOf$ to our probabilistic team logic this assumption is essential. Recall that the constant $1$ is defined by the unique $0$-ary function and is thus essentially included in the language. In structures of size at least $2$, the constant $0$ can be defined by $g(y)$ by the use of the formula
\[
\exists g \exists x \exists y \, (x\neq y \land g(x)=1). \footnote{$f(\tuple x)=0$ is always false for probability distributions $f$ in structures of size $1$.}
\]

In order to get some idea of the expressive power of  $\ESOf$, we note that  the uniformity of a distribution $f$ can be expressed with 
\[\phi(f)\dfn\forall \tuple x\tuple y (f(\tuple x)=0\vee f(\tuple y)=0 \vee f(\tuple x)=f(\tuple y)).\]
Furthermore, let $\frac{p}{q}$ be an arbitrary rational number. For $k\leq p$, denote by $\hat{k}$ the length $\log (p+1)$ bit sequence that encodes $k$, and denote by $\tuple y_{\hat{k}}$ the variable sequence obtained from $\hat{k}$ by replacing bits $0$ and $1$ with variables $y_0$ and $y_1$, respectively. For $l\leq q-p$, define 
 $\tuple z_{\hat{l}}$ analogously in terms of bit sequences of length $\log ((q-p)+1)$. For instance, $(y_0,y_0,\ldots ,y_0,y_0)$ is $\tuple y_{\hat{0}}$ and $(y_0,y_0,\ldots ,y_0,y_1)$ is $\tuple y_{\hat{1}}$. Let $E:=\{\tuple y_{\hat{k}}\tuple z_{\hat{0}}\mid 1\leq k\leq p\}\cup \{\tuple y_{\hat{0}}\tuple z_{\hat{l}}\mid 1\leq l\leq q-p\}$. Note that $E$ is not part of the syntax of our logic, but is used as a shorthand in the following formula.
Now $i(\tuple x)=\frac{p}{q}$ can be described by
\begin{align*}
\phi_{\frac{p}{q}}(\tuple x)\dfn &\exists y_0y_1\exists f \big(y_0\neq y_1 \wedge \bigwedge_{\tuple y\tuple z,\tuple y'\tuple z'\in E}f(\tuple y\tuple z)=f(\tuple y'\tuple z')\wedge \\
&\forall \tuple y\tuple z(\tuple y\tuple z \notin E \leftrightarrow f(\tuple y\tuple z)= 0)  \wedge i(\tuple x)= \SUM_{\tuple y}\tuple y\tuple z_{\hat{0}}\big).
\end{align*}
Note that, by construction, $E$ is finite, and consequently $\phi_{\frac{p}{q}}$ is an $\ESOf$-formula.

\section{Probabilistic Team Semantics}
In this section, we present probabilistic team semantics for probabilistic team logics. Before going to probabilistic semantics, we quickly review the basics of (multi)team semantics.

\subsection{Team and Multiteam Semantics}

Syntactically, team logics are extensions of  first-order logic $\FO$ given by the grammar rules:
\[
\phi ::=   x=y \mid x\neq y \mid R(\vec{x}) \mid \neg R(\vec{x}) \mid (\phi\land\phi) \mid (\phi\lor\phi) \mid \exists x\phi \mid \forall x \phi,
\]
where $\vec{x}$ is a tuple of first-order variables.

Let $D$ be a finite set of first-order variables and $A$ be a nonempty set. A function $s\colon D \to A$ is called an \emph{assignment}. The set $D$ is the \emph{domain} of $s$, and the set $A$ the \emph{codomain} of $s$. For a variable $x$ and $a\in A$, the assignment $s(a/x)\colon D\cup\{x\} \rightarrow A$ is equal to $s$ with the exception that $s(a/x)(x)=a$.

%%%%%%%%%%%%
%%%%%%%%%%%%
%%%%%%%%%%%%
%%%%%%%%%%%%
A $\emph{team}$ is a finite set of assignments with a common domain and codomain. Let $X$ be a team with codomain $A$, and let $F\colon X\to \Po(A)\setminus \{\emptyset\}$ be a function. We denote by $X[A/x]$ the modified team $\{s(a/x) \mid s\in X, a\in A\}$, and by $X[F/x]$ the team $\{s(a/x)\mid s\in X, a\in F(s)\}$. Let $\A$ be a $\tau$-structure and $X$ a team with codomain $A$, then we say that $X$ is a team of $\A$.

\begin{definition}
Let $\A$ be a
 $\tau$-structure and $X$ a team of $\A$. The satisfaction relation $\models_X$ for first-order logic is defined as follows:

\begin{tabbing}
left \= $\A \models_{X} (\psi \land \theta)$ \= $\Leftrightarrow$ \= $\forall s\in X: s(\tuple x) \in R^{\A}$\kill
%LIT
\> $\A \models_{X} x=y$ \> $\Leftrightarrow$ \> $\text{for all } s\in X: s(x)=s(y)$\\ 
\> $\A \models_{X} x\neq y$ \> $\Leftrightarrow$ \> $\text{for all }  s\in X: s(x) \not= s(y)$\\
\> $\A \models_{X} R(\tuple x)$ \> $\Leftrightarrow$ \> $\text{for all }  s\in X: s(\tuple x) \in R^{\A}$\\ 
\> $\A \models_{X} \neg R(\tuple x)$ \> $\Leftrightarrow$ \> $\text{for all } s\in X: s(\tuple x) \not\in R^{\A}$\\
%AND
\> $\A\models_{X} (\psi \land \theta)$ \> $\Leftrightarrow$ \> $\A \models_{X} \psi \text{ and } \A \models_{X} \theta$\\
%OR
\> $\A\models_{X} (\psi \lor \theta)$ \> $\Leftrightarrow$ \> $\A\models_Y \psi \text{ and } \A \models_Z \theta \text{ for some  $Y,Z\sub X$}$ s.t.\ $Y\cup Z = X$\\
%FORALL
\> $\A\models_{X} \forall x\psi$ \> $\Leftrightarrow$ \> $\A\models_{X[A/x]} \psi$\\
%EXISTS
\> $\A\models_{X} \exists x\psi$ \> $\Leftrightarrow$ \> $\A\models_{X[F/x]} \psi \text{ holds for some } F\colon X \to \Po(A)\setminus \{\emptyset\}$.
\end{tabbing}
\end{definition}

\emph{Multiteams} are multiset analogues of  teams. Below we give a short introduction to multiteam semantics, as defined by Durand~et~al.~\cite{foiks/0001HKMV16}, adjusted to the notation used later in this paper. 

\begin{definition}
A \emph{multiset} is a function $\mA\colon A \to \N$. The set $\{a\in A \mid \mA(a)\geq 1\}$ is the set of \emph{elements} of the multiset $\mA$, and $\mA(a)$ is the multiplicity of the element $a$.
A \emph{multiteam} is a multiset $\mX\colon X\to \N$ where $X$ is a team. The domain (codomain, resp.) of $\mX$ is defined as the domain (codomain, resp.) of $X$.
\end{definition}
For a multiset $\mA$, we define the \emph{canonical set representative} $\cset{\mA}$ of $\mA$ by 
\[
\cset{\mA} := \{\, (a,i) \mid a\in A, i\in \N, \ 0<i\leq \mA(a)\,\}.
\]
We say that a multiset $\mA$ is a submultiset of a multiset $\mB$, and write $\mA \subseteq \mB$, if and only if $\cset{\mA} \subseteq \cset{\mB}$. We write $\mA = \mB$ if and only if both $\mA \subseteq \mB$ and $\mB \subseteq \mA$ hold.
The \emph{disjoint union} $\mA \uplus \mB$ of $\mA$ and $\mB$ is the function $A\cup B\to\N$ defined by
\[
\mA \uplus \mB (s):=
\begin{cases}
\mA(s)+\mB(s) & \text{if $s\in A$ and $s\in B$},\\
\mA(s) & \text{if $s\in A$ and $s\not\in B$},\\
\mB(s) & \text{if $s\not\in A$ and $s\in B$}.
\end{cases}
\]

We write $\lvert \mA \rvert$ to denote the size of  $\mA$, i.e., $\lvert \mA \rvert := \sum_{a\in A} \mA(a)$.
Let $\mX$ be a multiteam, $A$ a finite set, and $F\colon \cset{\mX}\to \Po(A)\setminus \emptyset$ a function. We denote by $\mX[A/x]$ the modified multiteam defined as 
\[
\biguplus_{s\in X}\biguplus_{a\in A} \{\big(s(a/x), \mX(s)\big)\}.
\]
By $\mX[F/x]$ we denote the multiteam defined as
\[
\biguplus_{s\in X}\biguplus_{1\leq i \leq \mX(s)} \{\big(s(b/x), 1\big) \mid b\in F\big((s,i)\big)\}.
\]

A multiteam $\mX$ \emph{over} $\A$ is a multiteam with codomain $A$. 
We are now ready to define multiteam semantics for first-order logic. In the semantical clauses below, we use the lax semantics for existential quantifier and strict semantics for disjunction as defined by Durand~et.~al~\cite{foiks/0001HKMV16}.

%SEMANTICS OF FO-PART
\begin{definition}[Multiteam semantics]\label{laxmulti}
Let $\A$ be a %$\tau$-
 $\tau$-structure and $\mX$ a multiteam over $\A$. The satisfaction relation $\models_{\mX}$ is defined as follows:
\begin{tabbing}
left \= $\A \models_{\mX} (\psi \land \theta)$ \= $\Leftrightarrow$ \= $\forall s\in X: \text{ if } \mX(s)\geq 1 \text{ then }  s(\tuple x) \not\in R^{\A}$\kill
%LIT
\> $\A \models_{\mX} x=y$ \> $\Leftrightarrow$ \> $\text{for all } s\in X: \text{ if } \mX(s)\geq 1 \text{ then }  s(x)=s(y)$\\ 
\> $\A \models_{\mX} x\neq y$ \> $\Leftrightarrow$ \> $\text{for all } s\in X: \text{ if } \mX(s)\geq 1 \text{ then }  s(x)\not=s(y)$\\
\> $\A \models_{\mX} R(\tuple x)$ \> $\Leftrightarrow$ \> $\text{for all } s\in X: \text{ if } \mX(s)\geq 1 \text{ then }  s(\tuple x) \in R^{\A}$\\ 
\> $\A \models_{\mX} \neg R(\tuple x)$ \> $\Leftrightarrow$ \> $\text{for all } s\in X: \text{ if } \mX(s)\geq 1 \text{ then }  s(\tuple x) \not\in R^{\A}$\\
%AND
\> $\A\models_{\mX} (\psi \land \theta)$ \> $\Leftrightarrow$ \> $\A \models_{\mX} \psi \text{ and } \A \models_{\mX} \theta$\\
%OR
\> $\A\models_{\mX} (\psi \lor \theta)$ \> $\Leftrightarrow $\> $\A\models_{\mY} \psi  \text{ and } \A \models_{\mZ} \theta \text{ for some multisets}$
%\\
%\>\>\> 
$\text{$\mY,\mZ\subseteq \mX$ s.t. $\mX = \mY\uplus\mZ$.}$\\
%FORALL
\> $\A\models_{\mX} \forall x\psi$ \> $\Leftrightarrow$ \> $\A\models_{\mX[A/x]} \psi$\\
%EXISTS
\> $\A\models_{\mX} \exists x\psi$ \> $\Leftrightarrow$ \> $\A\models_{\mX[F/x]} \psi \text{ holds for some function}$
%\\
%\>\>\>
 $F\colon \cset{\mX} \to \Po(A)\setminus \emptyset$.
\end{tabbing}
\end{definition}

Using the counting measure,  a multiteam $\mX$ can be seen as a probability distribution over $X$; let $p_{\mX}$ denote the distribution defined as follows:
\[
p_{\mX} (s) := \frac{\mX(s)}{\sum_{t\in X} \mX(t)}.
\]

Conversely, every probability distribution $p$ over a team $X$ can be seen as a class $\mathcal{C}(p)$ of multiteams with that distribution as its counting measure:
\[
\mathcal{C}(p) := \{ \mX \mid p_{\mX}=p \}.
\]

Teams in $\mathcal{C}(p)$ can be seen as discrete approximations of the probability distribution $p$. In the section below we abandon the discrete approach and device team based logics that take probability distributions of teams as primitive. Intuitively, the semantics of these probabilistic logics is defined such that satisfaction of formulae with respect to probabilistic teams and their \emph{large enough} discrete approximations coincide.

%%%%%%%%%%%%%%%%%
\subsection{Probabilistic teams}

Let $D$ be a finite set of variables, $A$ a finite set, and $X$ a finite set of  assignments from $D$ to $A$. A $\emph{probabilistic team}$ $\X$ is a distribution $\X\colon X\rightarrow \rat$. We call $D$ and $A$ the variable domain and value domain of $\X$, respectively.
Let $\A$ be a $\tau$-structure and $\X$ a probabilistic team such that the domain of $\A$ is the value domain of $\X$.  Then we say that $\X$ is a probabilistic team of $\A$.
In the following, we will define two notations $\X[A/x]$ and $\X[F/x]$, similar to $\mX[A/x]$ and $\mX[F/x]$ of the previous section, in order to define the semantics of the universal and existential quantification of variables. 
Their intuition is depicted in Figure~\ref{fig:uni-exis}.
\begin{figure}[t]
	\centering
	\begin{tikzpicture}
		\foreach \x/\y in {0/0,0/1,0/2,2/0,2/1,2/2}{
			\draw[black,thick] (\x,\y) rectangle (\x+2,\y+1);
		}
		\foreach \y in {0.1,0.2,...,2.9}{
			\draw[black] (2,\y) -- (4,\y);
		}
		\foreach \y in {0,1,2}{
			\node at (1,\y+.5) {$s_\y$};
		}
		
		\node at (3,3.3) {$s_i(a/x)$};
		
		\draw [decorate,decoration={brace,amplitude=4pt,mirror},xshift=-4pt,yshift=0pt]
			(4.3,0.02) -- (4.3,0.98)node [black,midway,anchor=west,xshift=3pt] {$A\rightarrow\{\frac{1}{|A|}$\}};
		\draw [decorate,decoration={brace,amplitude=4pt,mirror},xshift=-4pt,yshift=0pt]
			(4.3,1.02) -- (4.3,1.98)node [black,midway,anchor=west,xshift=3pt] {$A\rightarrow\{\frac{1}{|A|}$\}};
		\draw [decorate,decoration={brace,amplitude=4pt,mirror},xshift=-4pt,yshift=0pt]
			(4.3,2.02) -- (4.3,2.98)node [black,midway,anchor=west,xshift=3pt] {$A\rightarrow\{\frac{1}{|A|}$\}};
	\end{tikzpicture}
	\qquad
	\begin{tikzpicture}
		\foreach \x/\y in {0/0,0/1,0/2,2/0,2/1,2/2}{
			\draw[black,thick] (\x,\y) rectangle (\x+2,\y+1);
		}
		\foreach \t in {.1,.3,.4,.7,1.2,1.5,1.8,1.9,2.09,2.4,2.55,2.9}{
			\draw[black] (2,\t) -- (4,\t); 
		}
		\foreach \y in {0,1,2}{
			\node at (1,\y+.5) {$s_\y$};
		}
		
		\node at (3,3.3) {$s_i(a/x)$};
		\draw [decorate,decoration={brace,amplitude=4pt,mirror},xshift=-4pt,yshift=0pt]
			(4.3,0.02) -- (4.3,0.98)node [black,midway,xshift=3pt,anchor=west] {$F(s_0)$};
		\draw [decorate,decoration={brace,amplitude=4pt,mirror},xshift=-4pt,yshift=0pt]
			(4.3,1.02) -- (4.3,1.98)node [black,midway,xshift=3pt,anchor=west] {$F(s_1)$};
		\draw [decorate,decoration={brace,amplitude=4pt,mirror},xshift=-4pt,yshift=0pt]
			(4.3,2.02) -- (4.3,2.98)node [black,midway,xshift=3pt,anchor=west] {$F(s_2)$};

	\end{tikzpicture}
	\caption{Intuition of universal quantification of $x$ (i.e., the set $\X[A/x]$) is depicted on the left side. The intuition of existential quantification of $x$ (i.e., the set $\X[F/x]$) is depicted of the right side. The height of a box labelled by an assignment corresponds to the assignments probability. E.g., on left the probability of $s_0$ is $\frac{1}{3}$ whereas the probability of $s_0(a/x)$ (for any $a\in A$) is $\frac{1}{3\lvert A\rvert}$.}\label{fig:uni-exis}
\end{figure}
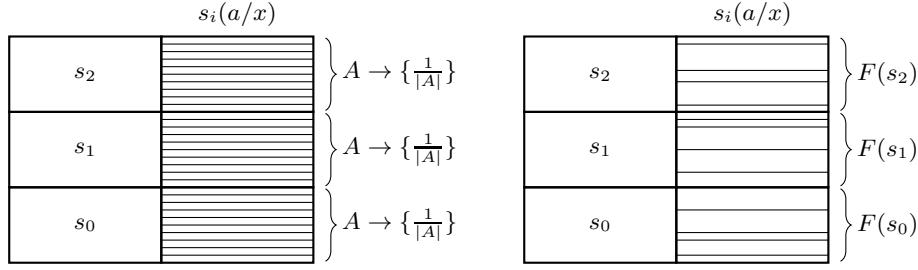

Let $\X\colon X\to\rat$ be a probabilistic team, $A$ a finite set, $p_A$ the set of all probability distributions  $d\colon A \to \rat$, and $F\colon X\to p_A$ a function.
We denote by $\X[A/x]$ the probabilistic team $X[A/x] \to \rat$  such that
\[
\X[A/x](s(a/x)) =  \sum_{\substack{t\in X \\ t(a/x) = s(a/x)}} \X(t) \cdot \frac{1}{\lvert A \rvert},
\]
for each $a\in A$ and $s\in X$. Note that if $x$ is a fresh variable then the righthand side of the above equation is simply $\X(s) \cdot \frac{1}{\lvert A \rvert}$.  By $\X[F/x]$ we denote the probabilistic team $X[A/x] \to \rat$ defined such that
\[
\X[F/x](s(a/x)) =  \sum_{\substack{t\in X \\ t(a/x) = s(a/x)}} \X(t) \cdot F(t)(a),
\]
for each $a\in A$ and $s\in X$. Again, if $x$ is a fresh variable, $\sum$ can be dropped from the above equation.

Let $\X\colon X\to \rat$ and $\Y\colon Y \to \rat$ be probabilistic teams with common variable and value domains, and let $k\in \rat$ be a rational number. We denote by $\X\kcup \Y$ the $k$-scaled union of $\X$ and $\Y$, that is, the probabilistic team $\X\kcup \Y\colon X\cup Y\to \rat$ defined such that for each $s\in X\cup Y$,
\[
(\X\kcup \Y)(s) :=
\begin{cases}
k\cdot \X(s) + (1-k)\cdot \Y(s) & \text{ if $s\in X$ and $s\in Y$},\\
k\cdot \X(s) & \text{ if $s\in X$ and $s\notin Y$},\\
(1-k)\cdot \Y(s) & \text{ if $s\in Y$ and $s\notin X$}.
\end{cases}
\]
We may now define probabilistic team semantics for first-order formulae.

\begin{definition}
Let $\A$ be a probabilistic $\tau$-structure over a finite domain $A$, and $\X\colon X \to \rat$ a proba\-bi\-listic team of $\A$. The satisfaction relation $\models_\X$ for first-order logic is defined as follows:

\begin{tabbing}
left \= $\A \models_{\X} (\psi \land \theta)$ \= $\Leftrightarrow$ \= $\forall s\in X: s(\tuple x) \in R^{\A}$\kill
%LIT
\> $\A \models_{\X} x=y$ \> $\Leftrightarrow$ \> $\text{for all } s\in X: \text{if }\X(s)> 0 \text{, then } s(x)=s(y)$\\ 
\> $\A \models_{\X} x\neq y$ \> $\Leftrightarrow$ \> $\text{for all } s\in X: \text{if }\X(s)> 0 \text{, then } s(x) \not= s(y)$\\
\> $\A \models_{\X} R(\tuple x)$ \> $\Leftrightarrow$ \> $\text{for all } s\in X: \text{if }\X(s)> 0 \text{, then } s(\tuple x) \in R^{\A}$\\ 
\> $\A \models_{\X} \neg R(\tuple x)$ \> $\Leftrightarrow$ \> $\text{for all } s\in X: \text{if } \X(s)> 0 \text{, then } s(\tuple x) \not\in R^{\A}$\\
%AND
\> $\A\models_{\X} (\psi \land \theta)$ \> $\Leftrightarrow$ \> $\A \models_{\X} \psi \text{ and } \A \models_{\X} \theta$\\
%OR
\> $\A\models_{\X} (\psi \lor \theta)$ \> $\Leftrightarrow$ \> $\A\models_\Y \psi \text{ and } \A \models_\Z \theta \text{ for some  $\Y,\Z,k$}$ s.t.\ $\Y\kcup \Z = \X$\\
%FORALL
\> $\A\models_{\X} \forall x\psi$ \> $\Leftrightarrow$ \> $\A\models_{\X[A/x]} \psi$\\
%EXISTS
\> $\A\models_{\X} \exists x\psi$ \> $\Leftrightarrow$ \> $\A\models_{\X[F/x]} \psi \text{ holds for some } F\colon X \to p_A $.
\end{tabbing}
\end{definition}

Next we define the semantics of probabilistic atoms considered in this paper: marginal identity and probabilistic independence atom. They were first introduced in the context of multiteam  semantics in \cite{foiks/0001HKMV16}. We define $|\X_{\tuple x =\tuple a}|$ where $\tuple x$ is a tuple of variables and $\tuple a$ a tuple of values, as the rational
\[|\X_{\tuple x =\tuple a}|:=\sum_{\substack{s(\tuple x)=\tuple a\\s\in X}} \X(s).\]
If $\phi$ is some first-order formula, then $|\X_{\phi}|$ is defined analogously as the total sum of weights of those assignments in $X$ that satisfy $\phi$.

If $\tuple x,\tuple y$ are variable sequences of length $k$, then $\tuple x \approx \tuple y$ is a \emph{marginal identity atom} with the following semantics:

\begin{equation}\label{eq:marg}
\A \models_{\X} \vec{x} \approx \vec{y} \Leftrightarrow \lvert {\X}_{\vec{x}=\vec{a}} \rvert =  \lvert {\X}_{\vec{y}=\vec{a}} \rvert \text{ for each $\vec{a}\in A^k$}
\end{equation}
Note that the equality  $\vert{\X}_{\vec{x}=\vec{a}} \rvert =  \lvert {\X}_{\vec{y}=\vec{a}} \rvert$ in \eqref{eq:marg} can be equivalently replaced with $\lvert{\X}_{\vec{x}=\vec{a}} \rvert 
\leq  \lvert {\X}_{\vec{y}=\vec{a}}\rvert$  since the tuples $\tuple a$ range over $A^k$. Due to this alternative formulation, marginal identity atoms were in \cite{foiks/0001HKMV16} called probabilistic inclusion~atoms.

If $\tuple x,\tuple y,\tuple z$ are variable sequences, then $\pcixyz$ is a \emph{probabilistic conditional independence atom} with the satisfaction relation defined as
\begin{align}
	\A\models_{\X} \pcixyz \label{def1}
\end{align}
if for all $s\colon\Var(\tuple x\tuple y\tuple z) \to A$ it holds that
\[
\lvert {\X}_{\tuple x\tuple y= s(\tuple x\tuple y)} \rvert \cdot  \lvert {\X}_{\tuple x\tuple z=s(\tuple x\tuple z)} \rvert = \lvert {\X}_{\tuple x\tuple y\tuple z=s(\tuple x\tuple y\tuple z)} \rvert \cdot \lvert {\X}_{\tuple x=s(\tuple x)} \rvert .
\]

The logic $\FOprob$ is now defined as the extension of $\FO$ with marginal identity and probabilistic conditional independence atoms. The following two examples demonstrate the expressivity of $\FOprob$.
\begin{example}\label{ex1}
The formula
\(
\forall \vec{y} \vec{x}\approx\vec{y}
\)
states that the probabilities for $\tuple x$ are uniformly distributed over all value sequences of length $|\tuple x|$.
\end{example}
\begin{example}\label{ex2}
We define a formula $\phi(x) \dfn \exists \alpha\beta\psi(x,\alpha,\beta)$ which expresses that the weight of a predicate $P(x)$ is at least two times that of a predicate $Q(x)$ in a probabilistic team over $x$. The subformula $\psi$ in $\phi$ is given as
\begin{align}
\psi &\dfn  x\alpha \approx x\beta\wedge \alpha=0\leftrightarrow \beta\neq 0\wedge
 \exists \gamma_{P}\gamma_Q \theta(x,\alpha,\beta,\gamma_{P}\gamma_Q), \text{ where} \label{eq:yks}\\
 \theta &\dfn \big((P(x)\wedge \alpha=0)\leftrightarrow \gamma_P=0\big)\wedge
Q(x)\rightarrow \gamma_Q=0 \wedge \gamma_P\approx \gamma_Q \label{eq:kaks}
\end{align}
 Now $\A\models_{\X}\phi(x)\iff |\X_{P(x)}|\geq 2\cdot |\X_{Q(x)}|$ for any $\X\colon X\to \rat$ where $\alpha$, $\beta$, $\gamma_{P}$, and $\gamma_Q$ are not in the variable domain of $\X$. The first two conjuncts in \eqref{eq:yks} indicate that the values of $\alpha$ must be chosen so that $\frac{1}{2}\cdot  |\Y_{P(x)}|=|\Y_{P(x)\wedge \alpha=0}|$. Where $\Y$ denotes the team obtained form $\X$ by evaluating the quantifiers $\exists \alpha\beta$. The first conjunct in \eqref{eq:kaks} implies  that  $|\Z_{P(x)\wedge \alpha=0}|=|\Z_{\gamma_P=0}|$ and the second that $|\Z_{Q(x)}|\leq |\Z_{\gamma_Q=0}|$, where $\Z$ is team obtained from $\Y$ by evaluating the quantifiers  $\exists \gamma_{P}\gamma_Q$. The third conjunct in \eqref{eq:kaks} then indicates that $|\Z_{\gamma_P=0}|=\lvert\Z_{\gamma_Q=0}|$. Put together, we have that 
 \begin{multline*}
 \lvert \X_{Q(x)}\rvert \overset{*}{=} |\Z_{Q(x)}|\leq |\Z_{\gamma_Q=0}|=|\Z_{\gamma_P=0}|=|\Z_{P(x)\wedge \alpha=0}|
 \overset{*}{=}|\Y_{P(x)\wedge \alpha=0}|=\frac{1}{2} |\Y_{P(x)}|\overset{*}{=}\frac{1}{2} |\X_{P(x)}|.
\end{multline*}
The equations $\overset{*}{=}$ follow from the fact that quantification of  fresh variables do not change the distribution of assignments with respect to the old variables.

\end{example}
Our next example relates probabilistic conditional independence atoms and marginal identity atoms to Bayesian networks.
A Bayesian network is a directed acyclic graph whose nodes represent random variables and edges represent dependency relations between these random variables. The applicability of Bayesian networks is grounded in the notion of conditional independence as the conditional independence relations encoded in the topology of such a network enable a factorization of the underlying joint probability distribution. Next we survey the possibility of refining Bayesian networks with information obtained from $\FOprob$ formulae.
\begin{example}\label{ex3}
Consider the Bayesian network $\mathbb{G}$ in Fig. \ref{network} that models beliefs about house safety using four Boolean random variables. We note that the awakening of  \texttt{guard} or \texttt{alarm} is conditioned upon both the presence of \texttt{thief} and \texttt{cat}. Furthermore, \texttt{cat} depends on \texttt{thief}, and \texttt{guard} and \texttt{alarm} are independent given \texttt{thief} and \texttt{cat}. From the network we obtain that the joint probability distribution for these variables can be factorized as 
 \begin{equation}\label{exeq}
 P(t,c,g,a)=P(t)\cdot P(c\mid t)\cdot P(g\mid t,c)\cdot P(a\mid t,c)
 \end{equation}
  where, e.g., $t$ abbreviates either $\texttt{thief}=T$ or $\texttt{thief}=F$, and $P(c\mid t)$ is the probability of $c$ given $t$. The joint probability distribution (i.e., a team $\X$) can hence be stored as in Fig. \ref{network}.

%%%
 Let $t,c,g,a$ now refer to random variables $\texttt{thief},\texttt{cat},\texttt{guard},\texttt{alarm}$. 
 The dependence structure of a Bayesian network is characterized by the so-called local directed Markov property  stating  that each variable  is conditionally independent of its non-descendants given its parents. For our network $\mathbb{G}$ the only non-trivial independence given by this property is $ \pci{tc}{g}{a}$. Hence a probabilistic team $\X$ over  $t,c,g,a$ factorizes according to \eqref{exeq} iff $\X$ satisfies $ \pci{tc}{g}{a}$. In this situation knowledge on various $\FOprob$ formulae can further improve the decomposition of the joint probability distribution. 
 Assume we have information suggesting that we may safely assume an $\FOprob$ formula $\phi$ on  $\X$:
\begin{itemize}
\item $\phi := t=F\to g=F$ indicates that \texttt{guard} never raises alert in absence of \texttt{thief}. In this case the two bottom rows of the conditional probability distribution for \texttt{guard} become superfluous.
\item $\phi := tca\approx tcg$ indicates that \texttt{alarm} and \texttt{guard} have the same reliability for any given value of \texttt{thief} and \texttt{cat}. Consequently, the conditional distributions for \texttt{alarm} and \texttt{guard} are equal and one of the them can be removed.
\item $\phi :=  \exists x(tcg\approx tcx \wedge \pmi{tcga}{y}\wedge x=T\leftrightarrow ay=TT)$ entails that \texttt{guard} is of a factor $P(y=T)$ less sensitive to raise alert than \texttt{alarm} for any given $\texttt{thief}$ and $\texttt{cat}$. The formula introduces a fresh free variable $y$, independent of any random variable in $\mathbb{G}$, and such that  that the probability of $ay=TT$ equals the probability of $g=T$ given $tc$. The latter property is expressed by introducing an auxiliary distribution for  $x$. In this case it suffices to store the conditional probability table for \texttt{alarm} and the probability $P(y=T)$.
\end{itemize}

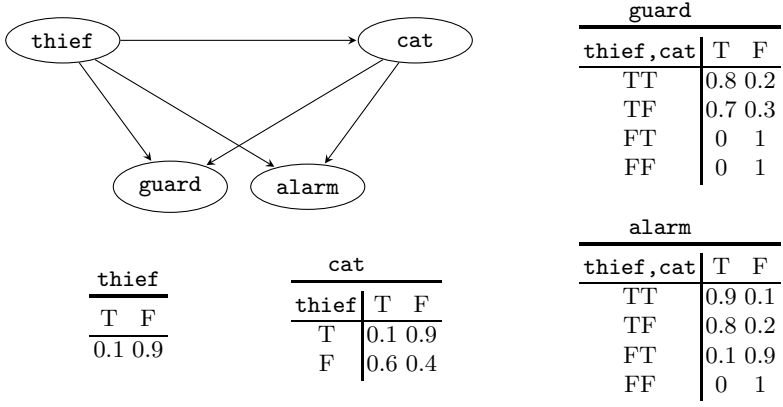
\begin{figure}[t]
\begin{tabular}{ccc}
\begin{tabular}{cc}
\multicolumn{2}{c}{
\begin{tikzpicture}[->,>=stealth, shorten >=.5pt,auto,node distance=2cm,
 main node/.style={circle,draw,ellipse,minimum size=0.6cm,minimum width=15mm,draw,font=\sffamily\bfseries}]
\usetikzlibrary{arrows}
\usetikzlibrary{shapes}
  \node[main node]  (1)  {\texttt{thief}};
 \node[main node, right=75pt]   (2) [right of=1] {\texttt{cat}};
   \node[main node, below=15pt]   (3) [below right  of=1] {\texttt{guard}};
   \node[main node, below=15pt]   (4) [below  left of=2] {\texttt{alarm}};
\path[every node/.style={sloped,anchor=south,auto=false}]
    (1) edge (2)
    (1) edge (3)
    (1) edge (4)
    (2) edge (4)
    (2) edge (3);
\end{tikzpicture}
}\\
{}\\
\mbox{\hspace{1cm}
\begin{tabular}{cc}
\multicolumn{2}{c}{\texttt{thief}}\\\toprule
T&F\\\hline
$0.1$&$0.9$
\end{tabular}}
&
\mbox{\hspace{1cm}
\begin{tabular}{c|cc}
\multicolumn{2}{c}{\texttt{cat}}\\\toprule
\texttt{thief}&T&F\\\hline
T&$0.1$&$0.9$\\
F&$0.6$&$0.4$
\end{tabular}}
\end{tabular}
&
\hspace{1cm}
&
\begin{tabular}{c}
\mbox{
\begin{tabular}{c|cc}
\multicolumn{2}{c}{\texttt{guard}}\\\toprule
\texttt{thief,cat}&T&F\\\hline
TT&$0.8$&$0.2$\\
TF&$0.7$&$0.3$\\
FT&$0$&$1$\\
FF&$0$&$1$\\
\end{tabular}}\\

{}\\

\mbox{
\begin{tabular}{c|cc}
\multicolumn{2}{c}{\texttt{alarm}}\\\toprule
\texttt{thief,cat}&T&F\\\hline
TT&$0.9$&$0.1$\\
TF&$0.8$&$0.2$\\
FT&$0.1$&$0.9$\\
FF&$0$&$1$\\
\end{tabular}}\\
\end{tabular}

\end{tabular}
\caption{Bayesian network $\mathbb{G}$ and its related conditional distributions\label{network}
}
\end{figure}
\end{example}

Next we connect probabilistic teams to multiteams. 
Denote by $\Prob$ the mapping that transforms a multiteam to its corresponding probabilistic team, i.e., given a multiteam $\mX$, $\Prob(\mX)$ is the probabilistic team $\X\colon X\to \rat$ such that
\[\X(s)=\frac{\mX(s)}{\sum_{s'\in X} \mX(s')}.\]

It follows from the definitions that $\Prob$ preserves the truth condition for marginal identity and probabilistic independence atoms.
\begin{proposition}\label{prop:equi}
Let $\phi$ be a marginal identity or a probabilistic independence atom, let $\mX$ be a multiteam of a structure $\A$, and let $\X$ be a probabilistic team of $\A$ such that $\X=\Prob(\mX)$.  Then 
\(\A\models_{\mX} \phi \iff \A \models_{\X} \phi\).
\end{proposition}

The restriction of a team $X$ to $V$ is defined as $X\upharpoonright V=\{s\upharpoonright V\mid s\in X\}$ where $s\upharpoonright V$ denotes the restriction of the assignment $s$ to $V$. The restriction of a probabilistic team $\X$ to $V$ is then defined as the probabilistic team $Q\colon X\upharpoonright V \to \rat$ where 
\[Q(s)= \sum_{s'\upharpoonright V= s} P(s').\] 
The following locality property  indicates that satisfaction of  $\phi \in \FOprob$  is determined by the restriction of a probabilistic team to the free variables of $\phi$. The set of \emph{free variables} $\Fr(\phi)$ of a formula $\phi\in \FOprob$ is defined recursively as in first-order logic with the addition that for probabilistic independence and marginal identity atoms $\phi$,  $\Fr(\phi)$ consists of all variables that appear in $\phi$.
\begin{proposition}[Locality]\label{prop:locality}
Let $\phi(\tuple x)\in \FOprob$ be a formula with free variables from $\tuple x=(x_1, \ldots ,x_n)$. Then for all structures $\A$ and probabilistic teams $\X:X\to \rat$ where $ \{x_1, \ldots ,x_n\} \subseteq V\subseteq\Dom(X)$,
\(\A\models_{\X} \phi \iff \A\models_{\X\upharpoonright V}\phi.\)
\end{proposition}
\begin{proof} For first-order atoms the claim is immediate. Furthermore, it is easy to check that the same holds for the atoms  $\vec{x} \approx \vec{y}$ and  $\pcixyz$ (for  multiteam semantics this has been discussed in \cite{foiks/0001HKMV16}). 

Assume then that $\phi:= \psi\vee \theta$, and that the claim holds for $\psi$ and $\theta$. Note first that for any  probabilistic teams $\X$ and    $\Y$ with common variable and value domains a simple calculation shows that  
\begin{equation}\label{localitydis}
 (\X\kcup \Y)\upharpoonright V = \X\upharpoonright V \kcup  \Y\upharpoonright V.     
 \end{equation}
Suppose that  $\A\models_{\X} \phi$.  Then there are $k$, $\Y$, and  $\Z$ such that $\X= \Y\kcup \Z$, $\A\models_{\Y} \psi$, and $\A\models_{\Z} \theta$. By the induction assumption, it holds that  $\A\models_{\Y \upharpoonright V} \psi$ and $\A\models_{\Z\upharpoonright V} \theta$. Now by \eqref{localitydis}, $\A\models_{\X\upharpoonright V} \phi$. The converse implication is proved analogously.  The proof is similar for the cases $\phi:= \exists x\psi$ and $\phi:= \forall x\psi$. \qed
\end{proof}

\section{Translation from $\FOprob$ to $\ESOf$}\label{sect:fromFOprob}
In this section, we show that any formula in $\FOprob$ can be equivalently expressed as a sentence of $\ESOf$ that has exactly one free function variable for encoding probabilistic teams. The following lemma will be used to facilitate the translation. This lemma has been shown by Durand~et~al.~\cite{foiks/0001HKMV16} for multiteams and accordingly, by Proposition~\ref{prop:equi}, it holds for probabilistic teams as well. The lemma entails that each probabilistic independence atom in $\phi\in \FOprob$ can be assumed to be either of the form $\pcixyz$ or of the form $\pci{\tuple x}{\tuple y}{\tuple y}$ for pairwise disjoint tuples $\tuple x,\tuple y,\tuple z$.
\begin{lemma}\cite{foiks/0001HKMV16}\label{rules}
Let $\A$ be a structure and $\X$ a probabilistic team over $\A$. Then
\begin{enumerate}[(i)]
\item $\A \models_{\X} \pci{\tuple x}{\tuple y}{\tuple z} \quad\Leftrightarrow\quad \A \models_{\X} \pcib{\tuple x}{\tuple y\setminus \tuple x}{\tuple z\setminus \tuple x}$, 
\item $\A \models_{\X} \pci{\tuple x}{\tuple y}{\tuple z} \quad\Leftrightarrow\quad \A \models_{\X} \pcib{\tuple x}{\tuple y\setminus \tuple z}{\tuple z\setminus \tuple y} \wedge \pcib{\tuple x}{\tuple y\cap\tuple z}{\tuple y\cap\tuple z}$.
\end{enumerate}
\end{lemma}

\begin{theorem}\label{thm:fromFOprob}
For every formula $\phi(\tuple x)\in \FOprob$ with free variables from $\tuple x=(x_1, \ldots ,x_n)$ there exists a formula $\phi^*(f)\in \ESOf$ with exactly one free function variable $f$ such that for all structures $\A$ and nonempty probabilistic teams $\X\colon X\to \rat$,
\[\A\models_{\X} \phi(\tuple x) \iff (\A,f_\X)\models \phi^*(f),
\]
where $f_\X:A^n\to \rat$ is the probability distribution such that $f_\X(s(\tuple x))=\X(s)$ for all  $s\in X$.
\end{theorem}
\begin{proof}
We give a compositional translation ${}^*$ from  $\FOprob$ to $\ESOf$. For a subsequence $\tuple x_i$ of $ \tuple x$, we denote by $\tuple x^c_i$ a sequence $\tuple x\setminus \tuple x_i$, and by $\tuple x(\tuple y/\tuple x_i)$ a sequence obtained from $\tuple x$ by replacing $\tuple x_i$ pointwise with $\tuple y$.
\begin{align}
%Rpos
\text{If }& \phi(\tuple x)\text{ is of the form }R(\vec{x_0})\text{, then }\phi^*(f):= \forall \vec{x} \big( f(\vec{x})=0 \lor R(\tuple x_0) \big).   \notag\\
%Rneg
\text{If }&\phi(\tuple x)\text{ is of the form }\neg R(\vec{x_0})\text{, then }\phi^*(f):= \forall \vec{x} \big( f(\vec{x})=0 \lor \neg R(\tuple x_0) \big).   \notag\\
%Approx
\text{If }& \phi(\tuple x)\text{ is }\tuple x_0 \approx \tuple x_1,
\text{then }\phi^*(f):=\forall \tuple z \, \SUM_{\vec{x}_0^c} f(\tuple x(\tuple z/\tuple x_0)) =  \SUM_{\vec{x}_1^c} f(\tuple x(\tuple z/\tuple x_1)). \notag\\ 
%PCIdisjoint
\text{If }& \phi(\tuple x)\text{ is }\pci{\tuple x_0}{\tuple x_1}{\tuple x_2}\text{ where }\tuple x_0,\tuple x_1, \tuple x_2\text{ are disjoint, then }
%\notag\\&
\phi^*(f) \dfn \forall \vec{x_0} \vec{x_1} \vec {x_2} \notag\\ 
&\SUM_{(\tuple x_0\tuple x_1)^c} f(\vec{x}) \times \SUM_{(\tuple x_0\tuple x_2)^c} f(\vec{x}) = 
%\notag\\
\SUM_{(\tuple x_0\tuple x_1\tuple x_2)^c}f(\vec{x}) \times \SUM_{\tuple x_0^c} f(\vec{x}). \notag\\
%PCIoverlap
\text{If }& \phi(\tuple x) \text{ is of the form }\pci{\tuple x_0}{\tuple x_1}{\tuple x_1}\text{ where }\tuple x_0,\tuple x_1\text{ are disjoint, then }\notag\\
&\phi^*(f):=\forall \tuple x_0\tuple x_1\big(\SUM_{(\tuple x_0\tuple x_1)^c}f(\tuple x)=0 \vee \SUM_{(\tuple x_0\tuple x_1)^c}f(\tuple x)=\SUM_{\tuple x_0^c}f(\tuple x)\big).\notag\\
%AND
\text{If }& \phi(\tuple x) \text{ is of the form } \psi_0(\tuple x)\land\psi_1(\tuple x)\text{, then }\phi^*(f) := \psi_0^*(f)\land\psi_1^*(f). \notag\\
%OR
\text{If }& \phi(\tuple x) \text{ is of the form }\psi_0(\tuple x)\lor\psi_1(\tuple x)\text{, then }\phi^*(f) := \psi^*_0(f) \lor \psi^*_1(f) \notag\\
&\lor \Big( \exists p g h k \big( \forall \vec{x} \forall y (y = l \lor y= r \lor (p(y)=0 \land k(\vec{x},y)=0 )) \label{eq:dis1}\\
&\; \land \forall \vec{x} ( k(\vec{x},l) = g(\vec{x}) \times p(l) \land k(\vec{x},r) = h(\vec{x}) \times p(r)) \label{eq:dis2}\\
&\; \land \forall \vec{x}\, (\SUM_y %(
k(\vec{x},y)= f(\vec{x}))
 \land \psi^*_0(g) \land \psi^*_1(h) \big) \Big). \label{eq:dis3}\\
%EXISTS
\text{If }& \phi(\tuple x) \text{ is } \exists y \psi(\tuple x,y)\text{, then } 
\phi^*(f) := \exists g \big( ( \forall \vec{x} \, \SUM_y g(\vec{x},y) = f(\vec{x})) \land \psi^*(g) \big). \notag\\
%FORALL
\text{If }& \phi(\tuple x) \text{ is of the form }\forall y \psi(\tuple x,y)\text{, then }\phi^*(f) :=\notag\\
 &\exists g \big(  \forall \vec{x} (\forall y \forall z  g(\vec{x},y) = g(\vec{x},z) \land \, \SUM_y g(\vec{x},y) = f(\vec{x})     ) \land \psi^*(g) \big).  \notag
\end{align}
The claim now follows via a straightforward induction on the structure of the formula. The cases for first-order and dependency atoms, and likewise for conjunctions, follow directly from the semantical clauses.

The case for disjunctions requires a bit more care. First note that $l$ (left) and $r$ (right) denote distinct constant symbols than can be defined by $\exists l \exists r \,l \not = r$ in the beginning of the translation ${}^*$.  Recall that a probabilistic team $\X$ satisfies a disjunction $(\phi \lor \psi)$ if and only if $\X$ satisfies either $\phi$ or $\psi$, or there exists two nonempty probabilistic teams $\Y$ and $\Z$ and a ratio $q\in\rat$ such that $\Y$ satisfies $\phi$, $\Z$ satisfies $\psi$, and, for each assignment $s$, it holds that $\X(s)= q \cdot \Y(s)+ (1-q) \cdot \Z(s)$. In the translation, we encode the value of $q$ by $p(l)$ and $(1-q)$ by $p(r)$. Line \eqref{eq:dis1} expresses that $p$ is such a function. We use $k(s(\vec{x}),l)$ and $k(s(\vec{x}),r)$ to encode the values of $q \cdot \Y(s)$ and $(1-q) \cdot \Z(s)$, respectively. Lines \eqref{eq:dis1} and \eqref{eq:dis2} together express that $k$ is such a function. Finally, the first part of line  \eqref{eq:dis3} expresses that $\forall s: \X(s)= q \cdot \Y(s)+ (1-q) \cdot \Z(s)$, whereas the latter part expresses that $\Y$ satisfies $\phi$, $\Z$ satisfies $\psi$.

The cases for the quantifiers follow directly by the semantical clauses.\qed
\end{proof}

\section{Translation from $\ESOf$ to $\FOprob$}\label{sect:fromESOf}
In this section, we construct a translation from $\ESOf$ to $\FOprob$. The proof utilises the observation that independence atoms and marginal identity atoms can be used to express multiplication and $\SUM$ in $\rat$, respectively. The translation then relates $\ESOf$ sentences in a certain normal form, presented in Lemma~\ref{lemma}, to open $\FOprob$ formulae.
Before this, we start by stating a lemma which expresses that existential quantification of a constant probability distribution $d$ can be characterised in $\FOprob$.  
Given a probabilistic team $\X\colon X\to \rat$, a tuple $\tuple x=(x_1, \ldots ,x_n)$ of fresh variables, and a probability distribution $d\colon A^n \to \rat$, we denote by $\X[d/\tuple x]$  the probabilistic team $\Y$ where
\(\Y(s(\tuple a/\tuple x))=\X(s)\cdot d(\tuple a)\)
for all $s\in X$.
\begin{lemma}\label{lem:yks}
Let $\phi(\tuple x) \dfn \exists \tuple y (\pmi{\tuple x}{\tuple y}\wedge \psi(\tuple x,\tuple y))$ be a $\FOprob$-formula with free variables from $\tuple x=(x_1, \ldots ,x_n)$. Then for all structures $\A$ and probabilistic teams $\X\colon X\to \rat$ where $\{x_1, \ldots ,x_n\}\subseteq \Dom(X)$, 
\[\A\models_{\X} \phi \iff \A\models_{\X[d/\tuple y]} \psi\text{ for some }d\colon A^{|\tuple y|} \to \rat.\]
\end{lemma}
\begin{proof}
By the locality principle (Prop. \ref{prop:locality}) $\A\models_{\X} \phi$ if and only if $\A\models_{\X\upharpoonright \{x_1, \ldots ,x_n\}} \phi$.
Likewise it is straightforward to check that, for $d\colon A^{|\tuple y|}\to \rat$
\[
\A\models_{\X[d/\tuple y]} \psi  \text{ if and only if } \A\models_{\X\upharpoonright \{x_1, \ldots ,x_n\}[d/\tuple y]} \psi,
\]
since $\X[d/\tuple y]\upharpoonright \{x_1, \ldots ,x_n, \vec{y}\} = {\X\upharpoonright \{x_1, \ldots ,x_n\}[d/\tuple y]}$.
Accordingly, we may assume without loss of generality, that $\Dom(X)=\{x_1, \ldots ,x_n\}$.

Now
$\A\models_{\X} \phi$ iff there is a function $F\colon X\to p_A$ such that $\A\models_{\Y}\pmi{\tuple x}{\tuple y}\wedge \psi(\tuple x,\tuple y)$ where $\Y:=\X[F/\tuple y]$. Furthermore,
\[
\A\models_{\Y}\pmi{\tuple x}{\tuple y} \text{ iff } \lvert {\Y}_{\tuple x\tuple y=s(\tuple x)\tuple a} \rvert=\lvert {\Y}_{\tuple x= s(\tuple x)} \rvert \cdot  \lvert {\Y}_{\tuple y=\tuple a} \rvert  \text{ for all $s\in X$ and $\tuple a \in A^n$}.
\]
Since $\Dom(X)=\{x_1, \ldots ,x_n\}$, the right-hand side of the above is equivalent to
\[
\X(s)\cdot F(s)(\tuple a)= \X(s)\cdot \lvert {\Y}_{\tuple y=\tuple a} \rvert \text{ for all $s\in X$ and $\tuple a \in A^n$}.
\]
This is equivalent with saying that $\X[F/\tuple y]=\X[d/\tuple y]$ for some distribution  $d\colon A^n\to \rat$. \qed
\end{proof}

Before proceeding to the translation, we construct the following normal form for $\ESOf$ sentences.

\begin{lemma}\label{lemma}
Every $\ESOf$ sentence $\phi$ is equivalent to a  sentence $\phi^*$ of the form
$\exists \tuple f \forall  \tuple x\theta$,
where $\theta$ is quantifier-free and such that   its second sort identity atoms are of the form $f_i(\tuple u\tuple v) = f_j(\tuple u)\times f_k(\tuple v)$ or $f_i(\tuple u) = \SUM_{\tuple v}f_j(\tuple u\tuple v)$ for distinct $f_i,f_j,f_k$ such that at most one of them is not quantified.

\end{lemma}
\begin{proof}
First we define for each second sort term $i(\tuple x)$ a special formula $\theta_i$ defined recursively using fresh function symbols $f_i$ as follows:
\begin{itemize}
\item If $i(\tuple u)$ is $ g(\tuple u)$ where $g$ is a function symbol, then $\theta_i$ is defined as $ f_i(\tuple u)= g(\tuple u)$. (We may intepret $g(\tuple u)$ as $\SUM_{\emptyset} g(\tuple u)$).
\item If $i(\tuple u\tuple v)$ is $j(\tuple u)\times k(\tuple v)$
, then $\theta_i$ is defined as $\theta_j\wedge \theta_k\wedge f_i(\tuple u\tuple v )=f_j(\tuple u)\times f_k(\tuple v)$.
\item If $i(\tuple u)$ is $\SUM_{\tuple v}j(\tuple u\tuple v)$, then $\theta_i$ is defined as $ \theta_j\wedge f_i(\tuple u)=\SUM_{\tuple v}f_j(\tuple u\tuple v)$.
\end{itemize}
The translation $\phi\mapsto \phi^*$  then proceeds recursively on the structure of $\phi$.
\begin{enumerate}[(i)]
\item If $\phi$ is $  i(\tuple u)=j(\tuple v)$, then $\phi^*$ is defined as $\exists \tuple f (f_i(\tuple u)=f_j(\tuple v)  \wedge \theta_i\wedge 
\theta_j)$ 
where $\tuple f$ is lists the function symbols $f_k$ for each subterm $k$ of $i$ or $j$. If $\phi$ is $  i(\tuple u)\neq j(\tuple v)$, the translation is analogous.
\item If $\phi$ is an atom or negated atom of the first sort, then $\phi^*:=\phi$.
\item If $\phi$ is $ \psi_0\circ\psi_1$ where $\circ\in \{\vee,\wedge\}$, $\psi^*_0$ is $\exists\tuple  f_0\forall \tuple x_0\theta_0$, and $\psi^*_1$ is $\exists\tuple  f_1\forall\tuple x_1\theta_1$, then $\phi^*_1$ is defined as $\exists \tuple f_0\tuple f_1\forall \tuple x_0\tuple x_1 (\theta_0\circ \theta_1)$.
\item If $\phi$ is $\exists y\psi$ where $\psi^*$ is $\exists\tuple  f\forall \tuple x\theta$, then $\phi^*$ is defined as $\exists g\exists\tuple  f\forall \tuple x\forall y(g(y)=0\vee \theta)$.
\item  If $\phi $ is $\forall y\psi$ where $\psi^* $ is $\exists\tuple  f\forall \tuple x\theta$, then $\phi^*$ is defined as \[\exists\tuple  f^*\exists \tuple f_{\rm id}\exists d\forall yy'\forall \tuple x( d(y)=d(y')\wedge \bigwedge_{f^*\in \tuple f^*} \SUM_{\tuple x} f^*(y,\tuple x)=d(y)\wedge \theta^*)\] 
where $\tuple f^*$ is obtained from $\tuple f$ by replacing each $f$ from $\tuple f$ with $f^*$ such that $\ar(f^*)=\ar(f)+1$, $\tuple f_{\rm id}$ introduces new function symbol for each multiplication in $\theta$, and $\theta^*$ is obtained by replacing all second sort identities $\alpha$ of the form $ f_i(\tuple u\tuple v) = f_j(\tuple u)\times f_k(\tuple v)$ with 
\[f_{\alpha}(y,\tuple u\tuple v)=d(y)\times f^*_i(y,\tuple u\tuple v) \wedge f_{\alpha}(y,\tuple u\tuple v)=  f^*_j(y,\tuple u)\times f^*_k(y,\tuple v)\]
 and  $f_i(\tuple u) = \SUM_{\tuple v}f_j(\tuple u\tuple v)$ with $f^*_i(y,\tuple u) = \SUM_{\tuple v}f^*_j(y,\tuple u\tuple v)$

\item If $\phi$ is $\exists f\psi$ where $\psi^*$ is $\exists\tuple  f\forall \tuple x\theta$, then $\phi^*$ is defined as $\exists f\psi^*$.
\end{enumerate}
It is straightforward to check that $\phi^*$ is of the correct form and equivalent to $\phi$. What happens in (v) is that instead of guessing for all $y$ some distribution $f_y$ with arity $\ar(f)$, we guess a single distribution $f^*$ with arity $\ar(f)+1$ such that $f^*(y,\tuple u)=\frac{1}{|A|}\cdot f_y(\tuple u)$ where $A$ is the underlying domain of the structure. This is described by the existential quantification of a unary uniform distribution $d$ such that for all fixed $y$, $\SUM_{\tuple u}f^*(y,\tuple u)$ is $d(y)$. Then note that $f_y(\tuple u)= g_y(\tuple u')\cdot h_y(\tuple u'')$ iff $ \frac{1}{|A|}\cdot f^*(y,\tuple u)=  g^*(y,\tuple u')\cdot  h^*(y,\tuple u'')$ iff $ d(y)\cdot f^*(y,\tuple u)=  g^*(y,\tuple u')\cdot  h^*(y,\tuple u'')$. For identities over $\SUM$, the reasoning is analogous.\qed
\end{proof}

\begin{theorem}\label{thm:fromESOf}
Let $\phi(p)\in \ESOf$ be a sentence of the form $\exists \tuple f \forall  \tuple x\theta $ where $\theta$ is a quantifier-free $\FOf$ formula in which each second sort equality atom is of the form $f_i(\tuple x_i) = f_j(\tuple x_j)\times f_k(\tuple x_k)$ or $f_i(\tuple x_i) = \SUM_{\tuple x_k}f_j(\tuple x_k\tuple x_j)$ for distinct $f_i,f_j,f_k$ from $\{f_1, \ldots ,f_n\}\cup \{p\}$. Then there is a formula $\Phi\in \FOprob$ such that for all structures $\A$ and probabilistic teams $\X := p^\A$,
\[\A\models_{\X} \Phi \iff (\A,p)\models \phi.
\]
\end{theorem}
\begin{proof}
We define $\Phi$ as
\begin{equation*}
\Phi \dfn \forall \tuple x \exists \tuple y_1 \ldots \tuple y_n (\Theta \wedge \Psi) 
\end{equation*}
  where $\tuple x=(x_1, \ldots ,x_m)$,  $\tuple y_i$ are sequences of variables of length $\ar(f_i)$, $\Theta$ is a compositional translation from $\theta$, and 
  \begin{equation}\label{eq2}
  \Psi \dfn  \bigwedge_{i=1}^n \pmi{\tuple x\tuple y_1 \ldots \tuple y_{i-1}}{\tuple y_i}.
  \end{equation}
By Lemma \ref{lem:yks} it suffices to show 
  that for  all  distributions $f_1, \ldots ,f_n$, subsets  $M\sub A^m$, and probabilistic teams $\Y=\X  [M/\tuple x]  [f_1/\tuple y_1]  \ldots   [f_n/\tuple y_n]$,
  \begin{equation}\A\models_{\Y} \Theta
  \text{ iff }(\A, p,f_1, \ldots ,f_n) \models \theta(\tuple a)\text{ for all }\tuple a\in M.
  \end{equation} We show the claim by structural induction on the construction of $\Theta$.
 \begin{enumerate}
 \item If $\theta$  is an atom of the first sort, it clearly suffices to let $\Theta =\theta$. 
  \item Assume $\theta$ is of the form $f_i(\tuple x_i) = f_j(\tuple x_j)\times f_k(\tuple x_k)$. Then $\Theta$ is defined as 
 \[\Theta \dfn \exists \alpha \beta \Big((\alpha =0 \leftrightarrow \tuple x_i =\tuple y_i)  \wedge (\beta =0 \leftrightarrow \tuple x_j\tuple x_k =\tuple y_j\tuple y_k)\wedge \tuple x\alpha\approx \tuple x\beta) \Big).
 \]
 Assume that $(\A, p,f_1, \ldots ,f_n) \models \theta(\tuple a)\text{ for any given }\tuple a\in M$. Then we have $f_i(\tuple a_i) = f_j(\tuple a_j)\cdot f_k(\tuple a_k)$. We define  functions $F_{\alpha},F_{\beta}\colon \Y\to \{0,1\}$ so that $F_{\alpha}(s)=0$ iff $s(\tuple x_i)=s(\tuple y_i)$, and $F_{\beta}(s)=0$ iff $s(\tuple x_j\tuple x_k)=s(\tuple x_j\tuple x_k)$. It suffices to show that $\A\models_{\Z} \tuple  x \alpha\approx \tuple x \beta$ where $\Z=\Y[F_{\alpha}/\alpha][F_{\beta}/\beta]$. By the construction of $\Z$, we have $|\Z_{\tuple x \alpha=\tuple a 0}|=|\Z_{\tuple x\tuple y_i=\tuple a\tuple a_i}|=
 |\Y_{\tuple x=\tuple a}|\cdot f_i(\tuple a_i)$. Similarly, and using the hypothesis, we have   $|\Z_{\tuple x \beta=\tuple a 0}|=|\Z_{\tuple x\tuple y_j\tuple y_k=\tuple a\tuple a_j\tuple a_k}|=
 |\Y_{\tuple x=\tuple a}|\cdot f_j(\tuple a_j)\cdot f_k(\tuple a_k)= |\Y_{\tuple x=\tuple a}|\cdot f_i(\tuple a_i)$. Furthermore, since we have $|\Z_{\tuple x \alpha=\tuple a 1}|=|\Y_{\tuple x=\tuple a}|\cdot (1-f_i(\tuple a_i))=|\Z_{\tuple x \beta=\tuple a 1}|$, it follows that $\A\models_{\Y}\Theta$.

 Assume $\A \models_{\Y} \Theta$, and let $\Z$ be the extension of $\Y$ to $\alpha,\beta$ where $Z_{\alpha=0}=Z_{\tuple x_i=\tuple y_i}$ and $\Z_{\beta=0}=\Z_{\tuple x_j\tuple x_k=\tuple y_j\tuple y_k}$. Then $\A\models_{\Z}\tuple x\alpha \approx \tuple x \beta$ since $|\Y_{\tuple x =\tuple a}|\cdot f_i(\tuple a_i)=
 |\Y_{\tuple x =\tuple a}|\cdot |\Y_{\tuple y_i=\tuple a_i}|=
  |\Y_{\tuple x\tuple y_i=\tuple a\tuple a_i}|=
  |\Y_{\tuple x\tuple x_i=\tuple a\tuple y_i}|=
  |\Z_{\tuple x\alpha=\tuple a 0}|=  |\Z_{\tuple x\beta=\tuple a 0}|=
  |\Y_{\tuple x\tuple x_j\tuple x_k=\tuple a \tuple y_j\tuple y_k}|=
  |\Y_{\tuple x\tuple y_j\tuple y_k=\tuple a \tuple a_j\tuple a_k}|=
    |\Y_{\tuple x=\tuple a }|\cdot  |\Y_{\tuple y_j= \tuple a_j}|\cdot     |\Y_{\tuple y_k=\tuple a_k}|=
        |\Y_{\tuple x=\tuple a }|\cdot f_j(\tuple a_j)\cdot     f_k(\tuple a_k)$
   for all $\tuple a\in M$. 
 
 \item Assume $\theta$ is of the form
 $f_i(\tuple x_i) = \SUM_{\tuple x_k}f_j(\tuple x_k\tuple x_j)$. 
 We define $\Theta$  as 
 \[\Theta \dfn \exists \alpha\beta\Big((\alpha=0\leftrightarrow \tuple x_i=\tuple y_i) \wedge (\beta=0 \leftrightarrow  \tuple x_j=\tuple y_j )\wedge \tuple x \alpha\approx \tuple x \beta\Big).
 \]
Assume that $(\A, p,f_1, \ldots ,f_n) \models \theta(\tuple a)\text{ for any given }\tuple a\in M$. Then  $f_i(\tuple a_i) = \SUM_{\tuple x_k}f_j(\tuple x_k\tuple x_j)$.
 We define  functions $F_{\alpha},F_{\beta}\colon \Y\to \{0,1\}$ such that $F_{\alpha}(s)=0$ iff $s(\tuple x_i)=s(\tuple y_i)$, and $F_{\beta}(s)=0$ iff $s( \tuple x_j)=s(\tuple y_j )$. Then $\A\models_{\Z} \tuple  x \alpha\approx \tuple x \beta$ because 
$|\Z_{ \tuple x \alpha=\tuple a0}|=| \Y_{\tuple x\tuple x_i=\tuple a \tuple y_i}|=|\Y_{\tuple x \tuple y_i=\tuple a \tuple a_i}|= |\Y_{\tuple x=\tuple a}|\cdot f_i(\tuple a_i)= |\Y_{\tuple x=\tuple a}|\cdot \SUM_{\tuple x_k}f_j(\tuple x_k\tuple a_j)=|\Y_{\tuple x=\tuple a}|\cdot |\Y_{\tuple y_j=\tuple a_j}|=|\Y_{\tuple x\tuple y_j=\tuple a\tuple a_j}|=|\Y_{\tuple x\tuple x_j=\tuple a\tuple y_j}|=|\Z_{\tuple x\beta=\tuple a 0}|$. Furthermore, since $|\Z_{ \tuple x \alpha=\tuple a1}|=|\Z_{\tuple x\beta=\tuple a 1}|$ it follows that $\A\models_{\Y}\Theta$.

Assume that $\A\models_{\Y}\Theta$, and let $\Z$ be the extension of $\Y$ to $\alpha,\beta$ where $Z_{\alpha=0}=Z_{\tuple x_i=\tuple y_i}$ and $\Z_{\beta=0}=\Z_{\tuple x_j=\tuple y_j}$. Analogously to the previous case, we obtain $\A\models_{\Z}\tuple x \alpha \approx \tuple x \beta$ since  $|\Y_{\tuple x =\tuple a}|\cdot f_i(\tuple a_i)=
  |\Z_{\tuple x\alpha=\tuple a 0}|=  |\Z_{\tuple x\beta=\tuple a 0}|=
     |\Y_{\tuple x\tuple y_j=\tuple a  \tuple a_j}|=
  |\Y_{\tuple x=\tuple a }|\cdot  |\Y_{\tuple y_j= \tuple a_j}|=
              |\Y_{\tuple x=\tuple a }|\cdot \SUM_{\tuple x_k} f_j(\tuple a_j)$
   for all $\tuple a\in M$. 

 \item Assume $\theta$ is $\theta_0\wedge \theta_1$. Then we let $\Theta \dfn \Theta_0\wedge \Theta_1$, and the claim follows by a straightforward argument.
 \item Assume $\theta$ is $\theta_0\vee \theta_1$. Then we let
 \[
 \Theta \dfn \exists z \Big(\pci{\tuple x}{z}{z}\wedge (\Theta_0\wedge z=0) \vee (\Theta_1 \wedge \neg z=0)\Big).
 \]
 Assume $(\A, p,f_1, \ldots ,f_n) \models \theta_0\vee\theta_1$ for all $\tuple a \in M$. Then we find $M_0\cup M_1 =M$, $M_0\cap M_1=\emptyset$, such that  $(\A, p,f_1, \ldots ,f_n) \models \theta_i$ for all $\tuple a \in M_i$. We define $F: Y\to p_A$ so that $F_z(s)=c_i$ if $s(\tuple x)\in M_i$; by $c_i$ we denote the distribution 
 \[c_i(a) \dfn \begin{cases}
 1\text{ if $a=i$,}\\
 0\text{ otherwise.}
 \end{cases}
 \]
 Letting  $\Z_i=\X[M_i/\tuple x][f_1/\tuple y_1]\ldots [f_n/\tuple y_n][c_i/z]$, it follows that $\Z=\Y[F/z] = \Z_0\kcup \Z_1$ for  $k=\frac{M_0}{M}$. By the induction hypothesis $\A\models_{\Z_i} \Theta_i$, and accordingly $\A\models_{\Z_i} \Theta_0 \wedge z_i$. Since $\A\models_{\Z} \pci{\tuple x}{z}{z}$, we obtain by Proposition \ref{prop:locality} that $\A\models_{\Y} \Theta$.
 
Assume $\A\models_{\Y}\Theta$, and let $F\colon Y\to p_A$ be such that $\A\models_{\Z}\pci{\tuple x}{z}{z}\wedge ((\Theta_0\wedge z=0) \vee (\Theta_1 \wedge \neg z=0))$ for $\Z=\Y[F/z]$. Consequently,  $\A\models_{\Z'_0}\Theta_0$ and $\A\models_{\Z'_1}\Theta_1$ where $k\Z'_0= \Z_{z=0}$ and $(1-k)\Z'_1=\Z_{z=1}$ for  $k=|\Z_{z=0}|$. Since $\Z$ satisfies $\pci{\tuple x}{z}{z}$, we have
furthermore that  either $\Z_{\tuple x=\tuple a}=\Z_{\tuple xz=\tuple a 0}$ or $\Z_{\tuple x=\tuple a}=\Z_{\tuple xz=\tuple a 1}$ for all $\tuple a\in M$. This entails that $\Z_{z=0}=\Z_{\tuple x\in M_0}$ for some $M_0\sub M$. Therefore, $\Z'_0= \frac{|M|}{| M_0|}(\X[M/\tuple x][f_1/\tuple y_1]\ldots [f_n/\tuple y_n])_{\tuple x\in M_0}=\X[M_0/\tuple x][f_1/\tuple y_1]\ldots [f_n/\tuple y_n]$. By the induction hypothesis, we then obtain $(\A, p,f_1, \ldots ,f_n) \models \theta_0$ for all $\tuple a \in M_0$, and by analogous reasoning that $(\A, p,f_1, \ldots ,f_n) \models \theta_1$ for all $\tuple a \in M\setminus M_0$. Consequently, $(\A, p,f_1, \ldots ,f_n) \models \theta$ for all $\tuple a\in M$ which concludes the proof. \qed 

 \end{enumerate}

\end{proof}

\section{Complexity of $\FO(\approx)$ in multiteams vs. probabilistic teams}
 One of the fundamental results in logics in team semantics state that, in contrast to dependence and independence logics that correspond to existential second-order logic (accordingly, $\NP$), the expressivity of inclusion logic equals only that of positive greatest fixed-point logic and thus $\Ptime$ over finite ordered models \cite{Galliani:2011,gh13,vaananen07}. In this section, we consider the complexity of $\FO(\approx)$ that can be thought of as a probabilistic variant of inclusion logic. We present a formula $\phi\in \FO(\approx)$ which captures an $\NP$-complete property of multiteams (the example works under both strict and lax semantics introduced by Durand~et~al.~ \cite{foiks/0001HKMV16}). The possibility of expressing similar properties in probabilistic teams is left open. It is worth noting that our reduction is similar to the ones presented for quantifier-free dependence and independence logic formulae under team semantics \cite{2010Jarmo,2015arXiv150301144D} (see also the recent survey on complexity aspects of logics in team semantics \cite{DKV}).  
 
The following example relates $\FO(\approx)$ to the exact cover problem, a well-known $\NP$-complete problem \cite{Garey:1990}. Given a collection 
$\mathcal{S}$ of subsets of a set $A$, an \emph{exact cover} is a subcollection 
$\mathcal{S}^{*}$ of 
$\mathcal{S}$ such that each element in $A$ is contained in exactly one subset in 
$\mathcal{S}^{*}$. 
\begin{example}
Consider an exact cover problem over $A=\{1,2,3,4\}$ and $\mathcal{S}=\{S_1=\{1,2,3\},S_2=\{2\},S_4=\{1,3,4\}\}$. We construct a multiteam $\mX$ as follows. The multiteam $\mX$, depicted in Fig. \ref{examplefig}, is a constant function mapping all assignments to $1$.
\begin{figure}
	\centering
\begin{tabular}{ccccc}
\multicolumn{5}{c}{Multiteam $\mX$} \\ \toprule
\texttt{element} & \texttt{set} & \texttt{left} & \texttt{right}&$\mX(s)$\\\hline
$0$ &  $S_1$ & $1$ & $2$&$1$\\
$0$ &  $S_1$ & $2$ & $3$&$1$\\
$0$ &  $S_1$ & $3$ & $1$&$1$\\\hdashline
$0$ &  $S_2$ & $2$ & $2$&$1$\\\hdashline
$0$ &  $S_3$ & $1$ & $3$&$1$\\
$0$ &  $S_3$ & $3$ & $4$&$1$\\
$0$ &  $S_3$ & $4$ & $1$&$1$\\\hline
$1$ & $0$ & $0$ & $0$&$1$\\
$2$ & $0$ & $0$ & $0$&$1$\\
$3$ & $0$ & $0$ & $0$&$1$\\
$4$ & $0$ & $0$ & $0$&$1$\\
\bottomrule
\end{tabular}
\quad
\begin{tabular}{ccccccc}
\multicolumn{5}{c}{Probabilistic team $\X$} \\ \toprule
\texttt{element} & \texttt{set} & \texttt{left} & \texttt{right}&$\X(s)$&$\Y$&$\Z$\\\hline
$0$ &  $S_1$ & $1$ & $2$&$1/10$&$1/2$&$1/2$\\
$0$ &  $S_1$ & $2$ & $1$&$1/10$&$1/2$&$1/2$\\\hdashline
$0$ &  $S_2$ & $2$ & $3$&$1/10$&$1/2$&$1/2$\\
$0$ &  $S_2$ & $3$ & $2$&$1/10$&$1/2$&$1/2$\\\hdashline
$0$ &  $S_3$ & $3$ & $1$&$1/10$&$1/2$&$1/2$\\
$0$ &  $S_3$ & $1$ & $3$&$1/10$&$1/2$&$1/2$\\\hline
$1$ & $0$ & $0$ & $0$&$1/10$&&$1$\\
$2$ & $0$ & $0$ & $0$&$1/10$&&$1$\\
$3$ & $0$ & $0$ & $0$&$1/10$&&$1$\\
$4$ & $0$ & $0$ & $0$&$1/10$&&$1$\\
\bottomrule\\
\end{tabular}
\caption{A multiteam $\mX$ and a probabilistic team $\X$ \label{examplefig}}
\end{figure}
For each element $i$ of a subset $S_j$, we create an assignment that maps \texttt{element} to $0$, \texttt{set} to $s_j$, \texttt{left} to $i$, and \texttt{right} to the next element in $S_j$ (under some ordering). Also, if $S_j=\{i\}$, then right is mapped to $i$. In our example case these assignments appear above the solid line of the multiteam $\mX$ in Fig. \ref{examplefig}. Furthermore, for each element $i$ of $A$ we create an assignment that maps $\texttt{element}$ to $i$ and all other variables to $0$.
The answer to the exact cover problem is then positive iff $\mX$ satisfies
\begin{equation}\label{eq:example}
\phi := \texttt{set}\neq 0 \vee (\texttt{element} \approx \texttt{left}  \wedge \texttt{set},\texttt{right} \approx \texttt{set},\texttt{left} ).
\end{equation}
Note that since $\phi$ consists only of variables and connectives, we do not need to concern structures; we write $\mX\models \phi$ instead of $\A \models_\mX\phi$.
Now $\mX\models \phi$ if and only if $\mZ\models \texttt{set}\neq 0$ and $\mY\models \texttt{element} \approx \texttt{left}  \wedge \texttt{set},\texttt{right} \approx \texttt{set},\texttt{left}$, for some $\mZ$, $\mY$ such that $\mZ\uplus\mY=\mX$. Note that any subset of the assignments above the solid line in Fig. \ref{examplefig} satisfy $\texttt{set}\neq 0$ and could be a priori assigned to $\mZ$.  Note also that all of the assignments below the solid line must be assigned to the team $\mY$.
Henceforth, the conjunct $\texttt{element} \approx \texttt{left}$ forces to select assignments from above the solid line to $\mY$ exactly one assignment for each element of $A$. Then $\texttt{set},\texttt{right} \approx \texttt{set},\texttt{left}$ enforces that this selection either subsumes a subset $S_i$ or does not intersect it at all. In the example case, we can select the segments that corresponds to  sets $S_1$ and $S_2$.

The same reduction does not work for probabilistic teams. The probabilistic team $\X$ in Fig. \ref{examplefig} corresponds to the exact cover problem defined over $A=\{1,2,3\}$ and $\mathcal{S}=\{S_1=\{1,2\},S_2=\{2,3\},S_4=\{3,1\}\}$. This instance does not admit an exact cover. However, for satisfaction of \eqref{eq:example} by $\X$, taking half weights of the upper part for $\Y$ and all the remaining weights for $\Z$, we have  $\A\models_{\Y}\texttt{set}\neq 0$ and $\A\models_{\Z}\texttt{element} \approx \texttt{left}  \wedge \texttt{set},\texttt{right} \approx \texttt{set},\texttt{left} $ where $\X=\Y\kcup\Z$ for $k=\frac{3}{10}$.
\end{example}

It is straightforward to generalise the previous example to obtain the following result.
\begin{corollary}
Data complexity of the quantifier-free fragment of $\FO(\approx)$ under multiteam semantics is $\NP$-hard. This remains true for very simple fragments as  $\texttt{set}\neq 0 \vee (\texttt{element} \approx \texttt{left}  \wedge \texttt{set},\texttt{right} \approx \texttt{set},\texttt{left})$ is such a formula for which model checking is hard for $\NP$.
\end{corollary}
The obvious brute force algorithm gives inclusion to $\NP$.
\begin{theorem}
Data complexities of $\FO(\approx)$ and the quantifier-free fragment of $\FO{(\approx)}$ under multiteam semantics are $\NP$-complete.
\end{theorem}

\section{Conclusion}
In this article, we have initiated a systematic study of probabilistic team semantics. Some features of our semantics have been discussed in the literature but the logic $\FOprob$ has not been studied before in the probabilistic framework. Probabilistic logics with team semantics have already been applied in the context of so-called Bell's Inequalities of quantum mechanics \cite{hpv14}. On the other hand,  our work is in part motivated by the study of implication problems of database and probabilistic dependencies. Independence logic has recently been used to give a finite axiomatisation for the implication problem of independence atoms (i.e., EMVD's) and inclusion dependencies \cite{HannulaK16}. It is an interesting open question to apply our probabilistic logic to analyse the implication problem of conditional independence statements whose exact complexity is still open  \cite{Gyssens2014628,wbw00}.

\vspace{-1mm}
\section*{Acknowledgements}
The second author was supported by  grant 3711702 of the Marsden Fund. The third author was supported by  grant 308712  of the Academy of Finland.
This work was supported in part by the joint grant by the DAAD (57348395) and the Academy of Finland (308099).
We also thank the anonymous referees for their helpful suggestions.

\bibliographystyle{splncs03}
%\bibliography{multisets}
%\end{document}

\end{document}